\documentclass[journal]{IEEEtran}
\usepackage{amsmath,graphicx,url,times,dsfont,amssymb}
\usepackage{amsmath, amssymb, amsopn}
\usepackage{subfigure}
\usepackage{acronym}
\usepackage{caption}
\usepackage{color}
\usepackage[usenames,dvipsnames]{xcolor}
\usepackage{amsthm}
\usepackage{algorithmic}

\usepackage[linesnumbered,ruled]{algorithm2e}

\newtheorem{theorem}{Theorem}

\acrodef{TDOA}{time difference of arrival}
\acrodef{STFT}{short-time Fourier transform}
\acrodef{TF}{transfer function}
\acrodef{RTF}{relative transfer function}
\acrodef{GSC}{generalized sidelobe canceler}
\acrodef{LCMV}{linearly constrained minimum-variance}
\acrodef{RCB}{robust Capon beamformer}
\acrodef{FBF}{fixed beamformer}
\acrodef{BM}{blocking matrix}
\acrodef{NC}{noise canceler}
\acrodef{TF-GSC}{transfer function-\ac{GSC}}
\acrodef{DD-GSC}{Data Driven-\ac{GSC}}
\acrodef{RKHS}{reproducing kernel Hilbert space}
\acrodef{VV-RKHS}{vector-valued reproducing kernel Hilbert space}
\acrodef{VVMR}{vector-valued manifold regularization}
\acrodef{DOA}{direction of arrival}
\acrodef{SNR}{signal to noise ratio}
\acrodef{PSD}{power spectral density}
\acrodef{WLS}{weighted least square}
\acrodef{AIR}{acoustic impulse response}
\acrodef{RTF}{relative transfer function}
\acrodef{ATF}{acoustic transfer function}
\acrodef{WGN}{white Gaussian noise}
\acrodef{CPSD}{cross power spectral density}
\acrodef{GMM}{Gaussian Mixture Model}
\acrodef{GCC}{generalized cross-correlation}
\acrodef{RMSE}{root mean square error}
\acrodef{DRR}{direct to reverberant ratio}
\acrodef{ML}{maximum likelihood}
\acrodef{MRL}{Manifold Regularized Localization}
\acrodef{KNN}{k-nearest neighbours}
\acrodef{DD-KNN}{Diffusion Distance k-nearest neighbours}
\acrodef{SVD}{singular value decomposition}
\acrodef{SMR}{Sequential Manifold Regularization}
\acrodef{MRL}{Manifold Regularization for Localization}
\acrodef{DDS}{Diffusion Distance Search}
\acrodef{EM}{Expectation Maximization}
\acrodef{MUSIC}{multiple signal classification}
\acrodef{ESPRIT}{estimation of signal parameters via rotational invariance}

\newcommand\numberthis{\addtocounter{equation}{1}\tag{\theequation}}
\newcommand{\argmin}{\operatornamewithlimits{argmin}}

\begin{document}

\title{Semi-Supervised Sound Source Localization Based on Manifold Regularization}
\author{Bracha~Laufer-Goldshtein~\IEEEmembership{Student Member,~IEEE}, Ronen~Talmon,~\IEEEmembership{Member,~IEEE} and Sharon~Gannot,~\IEEEmembership{Senior Member,~IEEE}
\thanks{Bracha~Laufer-Goldshtein and Sharon Gannot are with the Faculty of Engineering, Bar-Ilan University,
	Ramat-Gan, 5290002, Israel (e-mail: bracha\_gold@walla.com, Sharon.Gannot@biu.ac.il); Ronen Talmon is with the Department of Electrical Engineering, The Technion-Israel Institute of Technology, Technion City, Haifa 3200003, Israel, (e-mail: ronen@ee.technion.ac.il).}}

\maketitle

\begin{abstract}
Conventional speaker localization algorithms, based merely on the received microphone signals, are often sensitive to adverse conditions, such as: high reverberation or low \ac{SNR}. In some scenarios, e.g. in meeting rooms or cars, it can be assumed that the source position is confined to a predefined area, and the acoustic parameters of the environment are approximately fixed. Such scenarios give rise to the assumption that the acoustic samples from the region of interest have a distinct geometrical structure. In this paper, we show that the high dimensional acoustic samples indeed lie on a low dimensional manifold and can be embedded into a low dimensional space. Motivated by this result, we propose a semi-supervised source localization algorithm which recovers the inverse mapping between the acoustic samples and their corresponding locations. The idea is to use an optimization framework based on manifold regularization, that involves smoothness constraints of possible solutions with respect to the manifold. The proposed algorithm, termed \ac{MRL}, is implemented in an adaptive manner. The initialization is conducted with only few labelled samples attached with their respective source locations, and then the system is gradually adapted as new unlabelled samples (with unknown source locations) are received. Experimental results show superior localization performance when compared with a recently presented algorithm based on a manifold learning approach and with the \ac{GCC} algorithm as a baseline.
\end{abstract}


\section{Introduction and Motivation}
\label{sec:intro}
The problem of source localization has attracted the attention of many researchers during the last decades. Various applications rely on the recovery of the spatial position of an emitting source, such as: automated camera steering, teleconferencing and beamformer steering for robust speech recognition. For this reason, considerable amount of efforts have been devoted to investigate this field and a wide range of methods have been proposed over the years. Common to all localization approaches is the utilization of multiple microphone recordings to infer the spatial information. The fundamental challenge is to attain robust localization in poor conditions, i.e., in the presence of high reverberation and background noises.

Conventional localization approaches can be roughly divided into two main categories: single- and dual-step approaches. In the first class of algorithms, the source location is determined directly from the microphone signals. The most dominant member of this class is the \ac{ML} algorithm. The algorithm is derived by applying the \ac{ML} criterion to a chosen statistical model of the received signals. This optimization often involves maximization of the output power of a beamformer, steered to all potential source locations~\cite{stoica1990maximum,chen2002maximum,zhang2007maximum}. Another type of single-stage approaches is high resolution spectral estimation methods, such as the well-known \ac{MUSIC} algorithm~\cite{schmidt1986multiple}, and the \ac{ESPRIT} techniques~\cite{roy1989esprit}.

In the dual-step approaches category, the first stage involves \ac{TDOA} estimation from spatially separated microphone pairs. The classical method for \ac{TDOA} estimation is the \acf{GCC} algorithm introduced in the landmark paper by Knapp and Carter~\cite{Knapp1976}. The \ac{GCC} method relies on the assumption of a reverberant-free model such that the \ac{ATF}, which relates the source and each of the microphones, is a pure delay. However, this assumption does not hold in the presence of room reverberation, rendering a performance deterioration~\cite{champagne1996performance}. Consequently, improvements of the \ac{GCC} method for the reverberant case were proposed~\cite{brandstein1997robust,stephenne1997new,dvorkind2005time}.

In the second algorithmic stage, the noisy \ac{TDOA} estimates are combined to carry out the actual localization. Each \ac{TDOA} estimate is associated with an infinite set of source positions, lying on a half of an hyperboloid. The locus of the speaker can be recovered by intersecting the hyperboloid surfaces corresponding to the measurements of different pairs of microphones. However, the computation of a 3-dimensional hyperboloids intersection is a cumbersome task and tends to be sensitive to \ac{TDOA} estimation errors. In far-field regime the hyperboloid can be approximated by a cone, and linear intersection estimate can be applied~\cite{brandstein1997closed}. Another simplifying approach is to recast the hyperbolic equations into a spherical form, and apply the nonlinear least squares approach~\cite{huang2001real}.

All the prementioned methods utilize the spatial information conveyed by the received signals, but do not rely on any prior information about the enclosure in which the measurements are obtained. In some scenarios, e.g. in meeting rooms or cars, the source position is confined to a predefined area. It is reasonable to assume that representative samples from the region of interest can be measured in advance. Examining the structures and patterns characterizing the representative samples can be utilized for formulating a data-driven model which relates the measured signals to their corresponding source positions. The additional information may help to better cope with the challenges posed by reverberation and noise. So far, only few attempts were made to involve training information for performing source localization.

Deleforge and Horaud in~\cite{deleforge20122d}, discussed a 2-D sound localization scheme, in the binaural hearing context. Their central assumption is that the binaural observations lie on an intrinsic manifold which is locally linear. Accordingly, they proposed a probabilistic piecewise affine regression model, that learns the localization-to-interaural mapping and its inverse. In~\cite{deleforge2013variational,deleforge2015acoustic}, the authors have generalized the algorithm to deal with multiple sources using variational \ac{EM} framework.

In~\cite{xiong2015} the task of \ac{DOA} estimation was formulated as a classification problem and a learning-based approach was presented. They proposed to extract features from the \ac{GCC} vectors and use a multilayer perceptron neural network to learn the nonlinear mapping from such features to the \ac{DOA}.

Talmon et al.~\cite{talmon2012parametrization} introduced a supervised method based on \emph{manifold learning}, using \emph{diffusion kernels}. The main idea is specifying the fundamental controlling parameters of the \ac{AIR} using a manifold learning scheme. Assuming that the position of the source is the only varying degree-of-freedom of the system at hand, this process is capable of recovering the unknown source locations. The key point of the algorithm is to use an appropriate diffusion kernel with a specifically-tailored distance measure, that is capable of finding the underlying independent parameters, dominating the system. Talmon et al.~\cite{Talmon2011} have applied this method to a single microphone system with a \ac{WGN} input.

In~\cite{laufer2013relative} we adopted the paradigm of~\cite{Talmon2011} and adapted it to a more realistic setting where the source is a speech signal rather than a \ac{WGN} signal. The power spectral density of the speech signal is non-flat (as well as non-stationary). Hence, the spectral variations may blur the variations attributed to the different possible locations of the source. In order to mitigate this problem, we committed two major changes in the algorithm presented in~\cite{Talmon2011}: 1) a second microphone was added and 2) the feature vector, that was originally based on the correlation function has been replaced by a \ac{PSD}-based vector. It should be emphasized that in~\cite{Talmon2011} the feature vector was associated with the \ac{AIR}, whereas in~\cite{laufer2013relative} the feature vector relied on the \acf{RTF} which is the Fourier transform of the relative impulse response.

Though localization algorithms based on the diffusion framework were shown to perform well, their fundamental drawback is that they do not provide any guarantee for optimality. In general the diffusion-based methods are implemented by a dual-stage approach. First, a low dimensional embedding of the representative samples is recovered in an unsupervised manner. Second, the new representation is used to estimate the unknown locations based on the labelled samples. The septation into two stages where one is entirely unsupervised and the other is entirely supervised is not necessarily optimal. Moreover, the unlabelled data are not exploited for the estimation itself.

The significance of combining both labelled and unlabelled data, in the source localization context, should be farther emphasized. Classification and regression algorithms which rely on training data, are very popular in various applications, such as: text categorization, handwriting recognition, images classification and speech recognition. Nowadays, there exist a rich database for each of these tasks, with considerable amount of examples with true labellings. Thus, these problems are more usefully solved using fully supervised approaches. On the contrary, in the localization problem the training should fit to the specific acoustic environment in which the measurements are obtained, thus, we cannot create a general database that corresponds to all possible acoustic scenarios. Instead, the training set should be generated individually for each acoustic environment. To obtain labelled data, one needs to generate recordings in a controlled manner and calibrate each of them precisely. Generating a large amount of labelled data is a cumbersome and impractical process. However, unlabelled data is freely available since it can be collected whenever someone is speaking. This greatly motivates the use of semi-supervised approaches, which mostly rely on unlabelled data, for the source localization problem. Another motivation is related to the special characteristics of the acoustic environment. As will be further elaborated in the paper, the unlabelled data can be utilized for forming a data-driven model of the acoustic environment that is very useful for performing robust source localization.

To address the limitations of the previous diffusion-based approaches, and to better utilize the unlabelled data, we propose the \acf{MRL} algorithm. The method recovers the inverse mapping between the acoustic samples and their corresponding locations. The gist of the algorithm is based on the concepts of manifold regularization on a \acf{RKHS}, introduced by Belkin et al.~\cite{Belkin2006}. The idea is to extended the standard supervised estimation framework by adding an extra regularization term which imposes a smoothness constraint on possible solutions with respect to a data-driven model. The model is learned empirically by forming a data adjacency graph over both labelled and unlabelled training samples. In this approach, the estimated location relies not only on the labelled samples, but also on the unlabelled ones. Moreover, in order to efficiently utilize unlabelled samples received during runtime, we propose an adaptive implementation. The \ac{MRL} algorithm iteratively updates the system, based on the new information which becomes available while accumulating new unlabelled data. We compare the proposed algorithm, with the \acf{DDS} method, which is a diffusion-based algorithm. 
The discussion is supported by an experimental study based on simulated data.

The paper is organized as follows. In Section~\ref{sec:problem}, we formulate the problem in a general noisy and reverberant environment. We motivate the choice of the \ac{RTF} for forming a feature vector and describe how it can be estimated based on the microphone measurements. In Section~\ref{sec:algo}, we discuss the existence of an acoustic manifold and formulate an optimization problem which relies on a data-driven model computed based on both labelled and unlabelled data. This formulation leads to the \ac{MRL} algorithm which is sequentially adapted by the unlabelled data accumulated during runtime. We briefly describe our previous localization method based on the diffusion framework \cite{laufer2013relative} in Section~\ref{sec:diffusion}. Accordingly, we describe the derivation of the \ac{DDS} algorithm which conducts a neighbours' search using the diffusion distance as an affinity measurement between \acp{RTF}. In Section~\ref{sec:results}, we demonstrate the algorithms' performance by an extensive simulation study. A comparison between the \ac{MRL} and the \ac{DDS} algorithms is carried out in Section~\ref{sec:discuss}. Section~\ref{sec:conclusions} concludes the paper.

\section{Problem Formulation}
\label{sec:problem}
We consider a standard enclosure, e.g., a conference room or a car interior, with moderate reverberation time. A single source located at $\mathbf{p}=[p_x,p_y,p_z]^T$ generates an unknown speech signal $s(n)$, which is received by a pair of microphones. The received signals, denoted by $x(n)$ and $y(n)$, are contaminated by an additive stationary noise, and are given by:
\begin{align}
x(n)&=a_1(n,\mathbf{p})*s(n)+u_1(n) \\
y(n)&=a_2(n,\mathbf{p})*s(n)+u_2(n) \numberthis
\label{eq:mics}
\end{align}
where $n$ is the time index, $a_i(n,\mathbf{p}),~i=\{1,2\}$ are the corresponding \acp{AIR} relating the source at position $\mathbf{p}$ and each of the microphones and $u_i(n),~i=\{1,2\}$ are uncorrelated \ac{WGN} signals. Linear convolution is denoted by $\ast$. Each of the \acp{AIR} is composed of the direct path between the source and the microphone, as well as reflections from the surfaces characterizing the enclosure. Consequently, even in moderate reverberation conditions, the \ac{AIR} is typically modelled as a long FIR filter.

The purpose is to localize the speaker based on the current received microphone signals $x(n)$ and $y(n)$. We assume that we are also given a set of prerecorded representative samples from the region of interest. The training set is composed of $N$ samples of measured signals $\{\bar{x}_i(n),\bar{y}_i(n)\}_{i=1}^N$ from various positions within the specified region. Only $l$ samples among the set are labelled, i.e., their originating position $\bar{\mathbf{p}}_i$ is known. The rest $u=N-l$ samples are unlabelled, namely, their corresponding source locations are unknown. To summarize, the training set is composed of $l$ labelled examples $\{\bar{x}_i(n),\bar{y}_i,\bar{p}_i\}_{i=1}^l$ and $u$ unlabelled examples $\{\bar{x}_i(n),\bar{y}_i\}_{i=l+1}^N$.

We are interested in a realistic scenario, where the amount of labelled data is significantly smaller than the amount of unlabelled data which can be collected online. Our goal is to build an on-line system which is initially given a small amount of labelled data, and is gradually adapted as new unlabelled samples are acquired.

The first step is to define an appropriate feature vector that faithfully represents the characteristics of the acoustic path and is invariant to the other factors, i.e., the stationary noise and the varying speech signals.
An equivalent representation of \eqref{eq:mics} is given by~\cite{Gannot2001}:
\begin{align*}
y(n)&=h(n,\mathbf{p})*x(n)+v(n)\\
v(n)&=u_2(n)-h(n)*u_1(n)  \numberthis
\label{eq:RIR}
\end{align*}
where $h(n,\mathbf{p})$ is the relative impulse response between the microphones with respect to the source, satisfying $a_2(n,\mathbf{p})=h(n,\mathbf{p})*a_1(n,\mathbf{p})$. In~\eqref{eq:RIR}, the relative impulse response represents the system relating the measured signal $x(n)$ as an input and the measured signal $y(n)$ as an output.

For convenience, we represent~\eqref{eq:RIR} in the frequency domain. The Fourier transform of the relative impulse response, termed the \ac{RTF}, is obtained by:
\begin{multline}
H(k,\mathbf{p}) = \frac{S_{yx}(k,\mathbf{p})}{S_{xx}(k,\mathbf{p})-S_{u_1u_1}(k)} =\\\frac{S_{ss}(k)A_2(k,\mathbf{p})A_1^*(k,\mathbf{p})}{S_{ss}(k)|A_1(k,\mathbf{p})|^2}
=\frac{A_2(k,\mathbf{p})}{A_1(k,\mathbf{p})} \quad k=0,\ldots,D-1
\label{eq:rtf}
\end{multline}
where $H(k,\mathbf{p})$ is the \ac{RTF}, $S_{yx}(k,\mathbf{p})$ is the \ac{CPSD} between $y(n)$ and $x(n)$, $S_{xx}(k,\mathbf{p})$ is the \ac{PSD} of $x(n)$, $S_{u_1u_1}(k)$ is the \ac{PSD} of the noise in the first microphone $u_1(n)$, and $S_{ss}(k)$ is the \ac{PSD} of the source $s(n)$. $A_1(k,\mathbf{p})$ and $A_2(k,\mathbf{p})$ are the \acp{ATF} of the respective \acp{AIR}, and $k$ denotes a discrete frequency index. The choice of the value of $D$ should balance the tradeoff between the correspondence with the relative impulse response length (large value) and latency considerations (small value).

Since $A_1(k,\mathbf{p})$ and $A_2(k,\mathbf{p})$ are unavailable, we estimate the \ac{RTF} by:
\begin{equation}
\hat{H}(k,\mathbf{p}) \equiv \frac{\hat{S}_{yx}(k,\mathbf{p})}{\hat{S}_{xx}(k,\mathbf{p})} .
\end{equation}
Note that this estimator is biased since we neglect the \ac{PSD} of the noise $S_{u_1u_1}(k)$. Alternatively, unbiased estimators can be used, such as the \ac{RTF} estimator based on the non-stationarity of the speech signal~\cite{Gannot2001}. However, we are not concerned with robust estimation of the \ac{RTF} since we will show that the proposed method is insensitive to this type of errors. Accordingly, we define the feature vector $\mathbf{h}_\mathbf{p}=[\hat{H}(0,\mathbf{p}), \ldots ,\hat{H}(D-1,\mathbf{p})]^T$ as the concatenation of estimated \ac{RTF} values in the $D$ frequency bins. In practice, we discard high frequencies in which the ratio in \eqref{eq:rtf} is meaningless due to weak speech components. For the sake of clarity, we omit the dependency on the position, and denote the \ac{RTF} feature vector by $\mathbf{h}$.

\section{Manifold Regularization for Localization}
\label{sec:algo}

Our goal is to recover the target function which transforms each \ac{RTF} to its corresponding location, based on the training set comprised of both labelled and unlabelled samples. Finding such an inverse mapping is non-trivial due to the complex nonlinear relation between the high dimensional \acp{RTF} and the originating locations. To mitigate this problem we adopt the concepts of manifold regularization, introduced by Belkin et al.~\cite{Belkin2005,Belkin2006}, and present it in the light of the acoustic environment and, in particular, for the source localization problem at hand. It is important to note that, originally, the concepts of manifold regularization were implemented for classification, whereas, here, it is applied to the problem of source localization which is a regression problem.

Two guiding principles are in the core of the proposed method, that will be termed \acf{MRL}. First, instead of using complex variational calculus for estimating the target function, we assume that the function resides in a \acf{RKHS}. Due to the special characteristics of the functions belonging to the \ac{RKHS}, the problem can be formulated simply as a system of linear equations. Second, we incorporate geometrical considerations, i.e., we use the information implied by the intrinsic patterns observed in the set of \acp{RTF} to build a data-driven model. Then, the solution is constrained to behave smoothly with respect to this data-driven model, representing the intrinsic structure of the \acp{RTF}.

\subsection{The Acoustic Manifold}
\label{sec:manifold}
As mentioned in Section~\ref{sec:problem}, the \acp{RTF} have a high dimensional representation in $\mathbb{C}^D$ that corresponds to the vast amount of reflections from the different surfaces characterizing the enclosure. We assume that the \ac{RTF} samples, drawn from a specific region of interest in the enclosure, are not spread uniformly in the entire space of $\mathbb{C}^D$. Instead, they are confined to a compact manifold $\mathcal{M}$ of dimension $d$, which is much smaller compared to the dimension of the ambient space, i.e. $d\ll D$. This assumption is justified by the fact that the \acp{RTF} are influenced by only a small set of parameters related to the physical characteristics of the environment, such as: the enclosure dimensions and shape, the surfaces' materials and the positions of the microphones and the source. Moreover, we focus on a static configuration, in which the properties of the enclosure and the position of the microphones remain fixed. In such an acoustic environment, the only varying degree of freedom is the source location. Accordingly, we assume that the \acp{RTF} can be intrinsically embedded in a low dimensional manifold which is governed by the position of the source. The existence of such an acoustic manifold was discussed in detail in~\cite{laufer2015}, and was demonstrated with respect to the \ac{DOA} of the source. The main results will be briefly described in the experimental part, in Section~\ref{sec:analysis}.

Roughly, we consider a manifold of reduced dimensions which may have a complex nonlinear structure. However, in small neighbourhoods the manifold is locally linear, meaning that in the vicinity of each point it is flat and coincides with the tangent plane to the manifold at that point. Hence, the Euclidean distance can faithfully measure affinities between points that resides close to each other on the manifold. For larger scales, the Euclidean distance is meaningless, and we should rather use the geodesic distance on the manifold. However, the geodesic distance can be evaluated only when the structure of the manifold is known. In order to respect the manifold structure we will only examine local connections between points and disregard larger distances.

\subsection{Background of Reproducing Kernel Hilbert Spaces}
\label{sec:rkhs}
Our goal is to find the inverse-mapping function that receives an \ac{RTF} sample and returns the corresponding source location. In general, estimating a function that minimizes a cost function, is a cumbersome task that requires complex mathematical tools, such as variational calculus. One simplifying approach is to assume that the target function belongs to a certain class of functions with a specific structure. For example, it can be assumed that the target function belongs to a certain space of functions, spanned by an orthogonal basis. Hence, the target function can be represented by a linear combination of the basis functions, where the weights are determined according to the projections of the function on each of the basis functions. In our case we assume that the target function belongs to a \acf{RKHS} associated with a unique kernel function that evaluates each function in the space by an inner product. Rather than computing the basis functions spanning the space, we use an analogues representation with linear combinations of the kernel function. According to this representation, the problem can be converted to a simple linear estimation of a finite set of parameters.

We will first represent the kernel function and its properties, and then define the \ac{RKHS} and discuss its representation by the kernel function that will be used for deriving the optimization problem in Section~\ref{sec:Opti}. In Appendix~\ref{sec:appA}, we show that the eigenfunctions associated with the kernel form an orthogonal basis for the \ac{RKHS}, and discuss an analogue representation in terms of these basis functions.

As implied by its name, an \ac{RKHS} is associated with a kernel function $k:\mathcal{M}\times \mathcal{M}\rightarrow\mathbb{R}$ that measures a pairwise affinity between \acp{RTF}. The kernel function must satisfy the following two conditions:
\begin{enumerate}
  \item Symmetry: $k(\mathbf{h}_i,\mathbf{h}_j)=k(\mathbf{h}_j,\mathbf{h}_i) \quad \forall \mathbf{h}_i,\mathbf{h}_j \in \mathcal{M}$.
  \item Positive semi-definite: the $n\times n$ matrix $\mathbf{K}$ with $K_{ij}=k(\mathbf{h}_i,\mathbf{h}_j)$ is positive semi-definite, for any arbitrary finite set of points $\{\mathbf{h}_i\}_{i=1}^{n} \in \mathcal{M}$.
\end{enumerate}

Another essential requirement from the kernel is that it defines a notion of locality, determined with accordance to a scaling factor $\varepsilon_k$: for $\Vert\mathbf{h}_i-\mathbf{h}_j\Vert\ll\varepsilon_k$, $k(\mathbf{h}_i,\mathbf{h}_j)\rightarrow 1$, and for $\Vert\mathbf{h}_i-\mathbf{h}_j\Vert\gg\varepsilon_k$, $k(\mathbf{h}_i,\mathbf{h}_j)\rightarrow0$.
A common choice is to use a Gaussian kernel with variance $\varepsilon_k$:
\begin{equation}
k(\mathbf{h}_i,\mathbf{h}_j)=\exp\left\{-\frac{\Vert\mathbf{h}_i-\mathbf{h}_j\Vert^2} {2\varepsilon_k} \right\}.
\label{eq:gauss}
\end{equation}
Clearly, the Gaussian kernel is a symmetric positive semi-definite function, and satisfies the locality property.

The locality property is of major importance in our case, since the kernel receives \acp{RTF}, sampled from the manifold $\mathcal{M}$. As discussed above, the manifold is in general nonlinear and is assumed to be locally linear over small patches. Due to its property of locality, the kernel function constitutes an affinity measure that respects the manifold structure.

An \ac{RKHS}, denoted as $\mathcal{H}_k$, is a Hilbert space of functions, mapping each $\mathbf{h}\in \mathcal{M}$ to $\mathbb{R}$, which is associated with a kernel $k$. We skip the formal definition of an \ac{RKHS} (for details see~\cite{aronszajn1950theory,berlinet2011reproducing}). Instead, we state the two main properties of an \ac{RKHS}:
\begin{itemize}
  \item for all $\mathbf{h}\in \mathcal{M}$, $k_\mathbf{h}(\cdot)\in \mathcal{H}_k$
  \item \textbf{The reproducing property:} for all $f\in\mathcal{H}_k$ and $\mathbf{h}\in \mathcal{M}$, $\langle f(\cdot),k_\mathbf{h}(\cdot) \rangle =f(\mathbf{h})$
\end{itemize}
where for each $\mathbf{h}\in \mathcal{M}$ we define the real valued function $k_\mathbf{h}(\cdot)\equiv k(\mathbf{h},\cdot)$.
The first property simply states that the \ac{RKHS} consists of all functions defined by the kernel $k$ at some point on the manifold. The second property implies that the kernel $k$ has a special property that it evaluates all the functions in the space by an inner product. For example, in $l_2$ the delta function has the reproducing property since it evaluates all the functions in $l_2$: $\left\langle \delta(\mathbf{h},\cdot),f(\cdot) \right\rangle_{l_2} =f(\mathbf{h})$. However, this does not define an \ac{RKHS}, since the delta function does not belong to $l_2$.

We have seen that an \ac{RKHS} is associated with a unique reproducing kernel function. In the opposite direction, known as the Moore-Aronszajn theorem, every symmetric, positive definite kernel $k$ defines a unique RKHS $\mathcal{H}_k$ that is given by the \emph{completion} (an expansion that includes the limits of all Cauchy sequences) of the space of functions spanned by the set $\{k_{\mathbf{h}_i}(\cdot)\}$:
\begin{equation}
\{f|f(\cdot)=\sum_i a_i k_{\mathbf{h}_i}(\cdot); i\in\mathds{N},a_i\in\mathds{R}, \mathbf{h}_i\in \mathcal{M} \}
\label{eq:span}
\end{equation}
with respect to the following inner product:
\begin{align}
\left\langle f(\cdot),g(\cdot) \right\rangle &=\left\langle\sum_i a_i k_{\mathbf{h}_i}(\cdot),\sum_j b_j k_{\mathbf{h}_j}(\cdot) \right\rangle \\ \nonumber
&=\sum_{i,j}a_i b_j k(\mathbf{h}_i,\mathbf{h}_j).
\label{eq:product}
\end{align}
It can be easily verified that the two mentioned properties of an \ac{RKHS} are satisfied by this definition. Obviously, the reproducing kernel belongs to the space, and the reproducing property holds, since:
\begin{align}
\left\langle f(\cdot), k_\mathbf{h}(\cdot)\right\rangle &=\left\langle \sum_i a_i k_{\mathbf{h}_i}(\cdot),k_\mathbf{h}(\cdot) \right\rangle \\ \nonumber
&=\sum_i a_i k(\mathbf{h}_i,\mathbf{h})=f(\mathbf{h}).
\label{eq:prop}
\end{align}

Additional view of an \ac{RKHS}, based on Mercer's theorem~\cite{mercer1909functions}, is discussed in Appendix~\ref{sec:appA}. According to this view point, any function $f\in \mathcal{H}_k$ can be represented by an orthogonal basis of functions $\{\psi_i(\cdot)\}$ related to the kernel $k$:
\begin{equation}
\mathcal{H}_k = \{ f |f(\cdot) =\sum_i a_i\psi_i(\cdot) \mbox{ and } ||f||_{\mathcal{H}_k} < \infty\}.
\end{equation}
To circumvent the computation of the basis functions, we use the representation of \eqref{eq:span}, in terms of the kernel function.

\subsection{Optimization and Manifold Regularization}
\label{sec:Opti}
In this section we present the optimization over the target function assuming that it belongs to an \ac{RKHS} ${\mathcal{H}_k}$ with a reproducing kernel $k$. Formally, we search for a function $f_c:\mathbb{C}^D\rightarrow \mathbb{R}\quad c\in \{x,y,z\}$ which is the inverse mapping between an \ac{RTF} and its corresponding position, i.e. $f_c(\mathbf{h})=p_c$. In this paper we focus on estimating one position coordinate, thus, we omit the coordinate subscript. However, the analysis, the results and the algorithm described here can be naturally extended to estimating several coordinates.

The search will be formulated by the following optimization problem:
\begin{equation}
f^*=\argmin_{f\in \mathcal{H}_k}\frac{1}{l}\sum_{i=1}^l V(f(\bar{\mathbf{h}}_i),\bar{p}_i)+\gamma_k\|f\|_{\mathcal{H}_k}^2+\gamma_M\|f\|_\mathcal{M}^2 \label{eq:man_reg}
\end{equation}
where $\|\cdot\|_{\mathcal{H}_k}^2$ is the \ac{RKHS} norm that corresponds to the inner product defined in~\eqref{eq:product}, $\|\cdot\|_\mathcal{M}^2$ is the \emph{intrinsic} norm defined with respect to the manifold $\mathcal{M}$, and $\gamma_k, \gamma_M$ are scalar parameters. The optimization problem consists of three components. The first term is an empirical cost function defined over the labelled samples $\{\bar{\mathbf{h}}_i\}_{i=1}^l$. The function $V$ evaluates the extent of correspondence between the evaluations of the target function $f(\bar{\mathbf{h}}_i)$ and the true labels $\bar{p}_i$. In our case, we set the cost function to be the squared loss function $(\bar{p}_i-f(\bar{\mathbf{h}}_i))^2$. Note that while the $l_2$ norm is not suitable for comparing between \acp{RTF}~\cite{laufer2015}, it is a reasonable choice for evaluating localization quality.

The two last terms in~\eqref{eq:man_reg} are regularization conditions. Roughly, their role is to prevent the solution from overfitting to the labelled examples. The second term is the Tikhonov regularization which penalizes the \ac{RKHS} norm of the function to impose smoothness condition in $\mathcal{H}_k$.
The additional regularization term, defined by the last term in~\eqref{eq:man_reg}, was introduced by Belkin et al.~\cite{Belkin2006}. This is an intrinsic regularization that represents a smoothness penalty of the function with respect to the manifold $\mathcal{M}$.

One natural choice for the intrinsic norm is to measure the gradient of the function along the manifold, i.e., to measure the variability of the function with respect to small movements on the manifold. Since the manifold structure is unknown, this term should be approximated on the basis of both labelled and unlabelled samples. The training set $\{\bar{\mathbf{h}}\}_{i=1}^N$, which includes different realizations of possible acoustic paths, can be viewed as a discrete sampling of the manifold $\mathcal{M}$. The manifold can be empirically represented by a graph in which the training samples are the graph nodes, and the weights of the edges are defined according to an $N\times N$ adjacency matrix $\mathbf{W}$ between the samples:
\begin{equation}
W_{ij}=
\left\{
	\begin{array}{ll}
		\exp\left\{-\frac{\Vert\bar{\mathbf{h}}_i-\bar{\mathbf{h}}_j\Vert^2} {2\varepsilon_w} \right\}  & \mbox{if } \bar{\mathbf{h}}_j \in \mathcal{N}_i \mbox{ or } \bar{\mathbf{h}}_i \in \mathcal{N}_j \\
		0 & \mbox{otherwise}
	\end{array}
\right.
\label{eq:W}
\end{equation}
where $\mathcal{N}_j$ is a set consisting of the $d$ nearest-neighbours of $\bar{\mathbf{h}}_j$ among $\{\bar{\mathbf{h}}_i\}_{i=1}^N$.

The adjacency matrix $\mathbf{W}$ is used to form the graph Laplacain $\mathbf{L}$, by $\mathbf{L}=\mathbf{D}-\mathbf{W}$, where $\mathbf{D}$ is a diagonal matrix with $\mathbf{D}_{ii}=\sum_{j=1}^N \mathbf{W}_{ij}$. It can be shown, under certain conditions, that the graph Laplacian $\mathbf{L}$ converges to a differential operator on the manifold $\mathcal{M}$, as was discussed in detail in~\cite{belkin2004semi,Coifman2005,Hein2005}. Hence, the gradient of the function along the manifold can be approximated using the graph Laplacian. Accordingly, an intrinsic measure of data-dependent smoothness is given by: $\|f\|_\mathcal{M}^2=\mathbf{f}^T\mathbf{L}\mathbf{f}$, where $\mathbf{f}=\left[f(\bar{\mathbf{h}}_1),...,f(\bar{\mathbf{h}}_N)\right]$. Thus, the optimization problem~\eqref{eq:man_reg} can be recast as:
\begin{equation}
f^*=\argmin_{f\in \mathcal{H}_K}\frac{1}{l}\sum_{i=1}^l (\bar{p}_i-f(\bar{\mathbf{h}}_i))^2+\gamma_k\|f\|_{\mathcal{H}_K}^2+\gamma_M\mathbf{f}^T\mathbf{L}\mathbf{f} .\label{eq:man_reg2}
\end{equation}
Further insight can be obtained by the expansion of the intrinsic regularization:
\begin{align}
\mathbf{f}^T\mathbf{L}\mathbf{f}&=\sum_{i,j=1}^Nf(\bar{\mathbf{h}}_i)L_{ij}f(\bar{\mathbf{h}}_j)\nonumber \\
&=\sum_{i=1}^N\left(\sum_{j=1}^N W_{ij}-W_{ii}\right)f^2(\bar{\mathbf{h}}_i)-\sum_{\substack{i,j=1 \\ i\neq j}}^NW_{ij}f(\bar{\mathbf{h}}_i)f(\bar{\mathbf{h}}_j) \nonumber\\
&=\sum_{i,j=1}^N W_{ij}f^2(\bar{\mathbf{h}}_i)-\sum_{i,j=1}^NW_{ij}f(\bar{\mathbf{h}}_i)f(\bar{\mathbf{h}}_j) \nonumber \\
&=\frac{1}{2}\sum_{i,j=1}^N W_{ij}\left(f(\bar{\mathbf{h}}_i)-f(\bar{\mathbf{h}}_j)\right)^2
\label{reg}
\end{align}
Intuitively, in \eqref{reg}, large $W_{ij}$, corresponding to strong similarity between $\bar{\mathbf{h}}_i$ and $\bar{\mathbf{h}}_j$, implies a tendency of $f(\bar{\mathbf{h}}_i)$ and $f(\bar{\mathbf{h}}_j)$ to be close to each other. For this reason, a truncated kernel was chosen in \eqref{eq:W}, since it is reasonable to penalize the function only when the corresponding \acp{RTF} resides in the same local neighbourhood.

Note that~\eqref{eq:man_reg2} is a semi-supervised formulation, since it involves both labelled and unlabelled samples. While the first term is merely based on the labelled samples, the last two terms are based on both labelled and unlabelled data. The two regularization parameters $\gamma_k$ and $\gamma_M$ balance between maximizing the correspondence to the labelled data, and maintaining low-complexity of possible solutions. In some respects, both regularization terms try to relate the target function to the manifold $\mathcal{M}$ by the two different kernels defined in~\eqref{eq:gauss} and~\eqref{eq:W}. Involving two kernels associated with different scales represents two different measurements of smoothness with respect to the manifold. Since the real structure of the manifold is unknown, the combination of both kernels is essential for obtaining a more accurate modelling of the manifold.

The Representer theorem~\cite{scholkopf2001generalized} states that the minimizer $f^*$ of \eqref{eq:man_reg2} is a linear combination of the kernel functions only in the set of labelled and unlabelled points $\{\mathbf{h}_i\}_{i=1}^N$, i.e., it is given by:
\begin{equation}
f^*(\mathbf{h})=\sum_{i=1}^Na_i k(\bar{\mathbf{h}}_i,\mathbf{h}) \label{eq:luexpen}
\end{equation}
where $\{a_i\}$ are the interpolation weights. In Appendix~\ref{sec:appB} we provide the proof of the theorem~\cite{Belkin2006}, which is derived by a simple orthogonality argument, and relies on  the specific structure of the functions in $\mathcal{H}_k$ implied by \eqref{eq:span}, together with the reproducing property that uniquely characterizes the \ac{RKHS}. The Representer theorem dramatically simplifies the regularized optimization problem of~\eqref{eq:man_reg2} so it can be formulated as a linear optimization over a finite set of parameters $\{a_i\}$.

\subsection{Derivation of the Localization Algorithm}
In the previous section we formulated an optimization problem with manifold regularization for recovering the target function $f$ in \eqref{eq:man_reg2}. Based on the Representer theorem stated in \eqref{eq:luexpen}, the optimization boils down to estimating the interpolation weights $\{a_i\}$. Substituting \eqref{eq:luexpen} in \eqref{eq:man_reg2} yields a second-order polynomial objective function of $\mathbf{a}=\left[a_1,...,a_N\right]^T$:
\begin{multline}
\mathbf{a}^*=\argmin_{\mathbf{a}\in \mathbb{R}^N}\frac{1}{l}
\left(\mathbf{q}-\mathbf{J}\mathbf{K}\mathbf{a}\right)^T\left(\mathbf{q}-\mathbf{J}\mathbf{K}\mathbf{a}\right)\\
+\gamma_k\mathbf{a}^T\mathbf{K}\mathbf{a}
+\gamma_M\mathbf{a}^T\mathbf{K}\mathbf{L}\mathbf{K}\mathbf{a} \label{eq:man_sol}
\end{multline}
where $\mathbf{K}$ is the $N \times N$ Gram matrix of $k$ defined by $K_{ij}=k(\bar{\mathbf{h}}_i,\bar{\mathbf{h}}_j)$; $\mathbf{I}_N$ is the $N\times N$ identity matrix; $\mathbf{J}$ is a $N \times N$ diagonal matrix: $\mathbf{J}=\mbox{diag}(1,...,1,0,...,0)$ with $l$ ones and $u$ zeros on its diagonal (functions as an indicator for the labelled samples in the set); and $\mathbf{q}=[\bar{p}_1,...,\bar{p}_l,0,...,0]^T$ is a label vector comprising the $l$ known positions of the labelled samples with $q_i=0$, for all $i>l$.
Differentiating with respect to $\mathbf{a}$ and comparing to zero, yields:
\begin{equation}
\frac{1}{l}\left(\mathbf{q}-\mathbf{J}\mathbf{K}\mathbf{a}\right)^T\left(-\mathbf{J}\mathbf{K}\right)\\
+\left(\gamma_k\mathbf{K}+\gamma_M\mathbf{K}\mathbf{L}\mathbf{K}\right)\mathbf{a}=0 \label{eq:diff}
\end{equation}
By rearranging \eqref{eq:diff}, we obtain the following linear system:
\begin{equation}
    \left[\mathbf{J} \mathbf{K}+l\gamma_k \mathbf{I}_N+ l\gamma_M \mathbf{L} \mathbf{K}\right]\mathbf{a}=\mathbf{q}.
    \label{eq:lineq}
\end{equation}
Accordingly, the interpolation weights $\mathbf{a}$ are given by:
\begin{equation}
\mathbf{a}^*=\left[\mathbf{J} \mathbf{K}+l\gamma_k \mathbf{I}_N+ l\gamma_M \mathbf{L} \mathbf{K}\right]^{-1}\mathbf{q}.    \label{eq:weights}
\end{equation}

Thus far, the computations were carried out offline based only on the training set, composed of both labelled and unlabelled samples. The input to the algorithm is a new pair of measurements $\{x(n),y(n)\}$, generated by an unknown source from an unknown location on the manifold. The corresponding feature vector $\mathbf{h}$ is estimated according to \eqref{eq:rtf}. The kernel between the new sample $\mathbf{h}$ and each of the training samples $\{\bar{\mathbf{h}}_i\}_{i=1}^N$, is evaluated. The position of the new measurement is estimated according to \eqref{eq:luexpen} by a weighted sum of these kernel evaluations multiplied by the weights given by \eqref{eq:weights}:
\begin{equation}
    \hat{p}=f(\mathbf{h})=\sum_{i=1}^Na_i^* k(\bar{\mathbf{h}}_i,\mathbf{h})
    \label{eq:estimate}
\end{equation}

\subsection{Adaptive Manifold Regularization for Localization}
\label{sec:SMR}
In this section we summarize the algorithm and formulate it in a dual-stage structure. We will take advantage of the fact that the optimization is derived in a semi-supervised manner, and propose an adaptive version. The algorithm is composed of two main parts: system adaptation and localization. In the adaptation stage, the interpolation weights $\mathbf{a}^*$ are computed according to~\eqref{eq:weights} based on the labelled and unlabelled samples, which were collected up to this point in time. In the localization stage, we receive a new pair of measurements $\{x(n),y(n)\}$ of an unknown source from an unknown location, and estimate the corresponding position based on the weights computed in the previous stage. The system is initialized with a small amount of labelled data, and after several iterations of the localization stage, the new unlabelled samples received during runtime, are utilized for system adaptation. Note that the adaptation process can potentially adjust to changes in the environmental conditions. However, this attribute was not examined in the current paper that focuses on static configurations. Examining dynamic scenarios with changing environmental conditions is left for future work.

The proposed \ac{MRL} algorithm is summarized in Algorithm~\ref{alg:algorithm1} and is illustrated in a flow diagram in Fig.~\ref{fig:flow}. The flow diagram emphasizes the duality between the two parts of the algorithm and the interaction between them. In the downward direction, the model of the system derived in the adaptation part is utilized for localization. In the upward direction, the new unlabelled samples acquired in the localization stage, are propagated and utilized for system adaptation. Moreover, note that the two rightmost (blue) blocks are semi-supervised whereas the rest of the blocks are unsupervised.

It should be emphasized that we do not present an update mechanism, but instead the weights are computed from scratch in each adaptation iteration. The development of a recursive version of the algorithm is left for future work.

The number of localization iterations between two successive adaptations is chosen empirically to obtain satisfactory performance. Note that if we choose a small value, increasing computational complexity, we will not gain much performance improvement. Adding only a small amount of unlabelled information do not change the weights significantly.

\begin{algorithm}
\caption{Manifold Regularization for Localization}

\begin{algorithmic}
\STATE \underline{System Adaptation}:\\
\SetKwInOut{Input}{Input}
\SetKwInOut{Output}{Output}
\Input{$N=l+u$ training points: $l$ labelled samples $\{\bar{x}_i(n),\bar{y}_i(n),\bar{p}_i\}_{i=1}^l$ and $u$ unlabelled samples $\{\bar{x}_i(n),\bar{y}_i(n)\}_{i=l+1}^N$  }
\Output{Interpolation weights $\mathbf{a}^*$}
\STATE \begin{enumerate}
         \item For each point estimate the corresponding \ac{RTF} $\bar{\mathbf{h}}_i$ according to \eqref{eq:rtf}.
         \item Construct the reproducing kernel matrix $\mathbf{K}$ and the adjacency matrix $\mathbf{W}$, according to \eqref{eq:gauss} and \eqref{eq:W} respectively, based on $\left\{\bar{\mathbf{h}}_i\right\}_{i=1}^N$.
         \item Compute the expansion weights $\mathbf{a}^*$ according to \eqref{eq:weights}.
       \end{enumerate}
\STATE
\STATE \underline{Localization:}\\
\Input{A new pair of measurements $\{x(n),y(n)\}$ produced by an unknown source from an unknown location}
\Output{Estimated position $\hat{p}$}
\STATE \begin{enumerate}
         \item Estimate the corresponding \ac{RTF} $\mathbf{h}$ according to \eqref{eq:rtf}.
         \item Compute the affinity between $\mathbf{h}$ and each of $\left\{\bar{\mathbf{h}}_i\right\}_{i=1}^N$, using the reproducing kernel.
         \item Estimate the new point location using the estimated interpolation weights: $\hat{p}=f(\mathbf{h})=\sum_{i=1}^Na_i^* k(\bar{\mathbf{h}}_i,\mathbf{h})$ .
       \end{enumerate}
\STATE After a several number of newly acquired samples, return to System Adaptation and add the new unlabelled samples.
\end{algorithmic}
\label{alg:algorithm1}
\end{algorithm}

\begin{figure*}[ht!]
\centering
\includegraphics[width=0.65\textwidth,height=0.33\textheight]{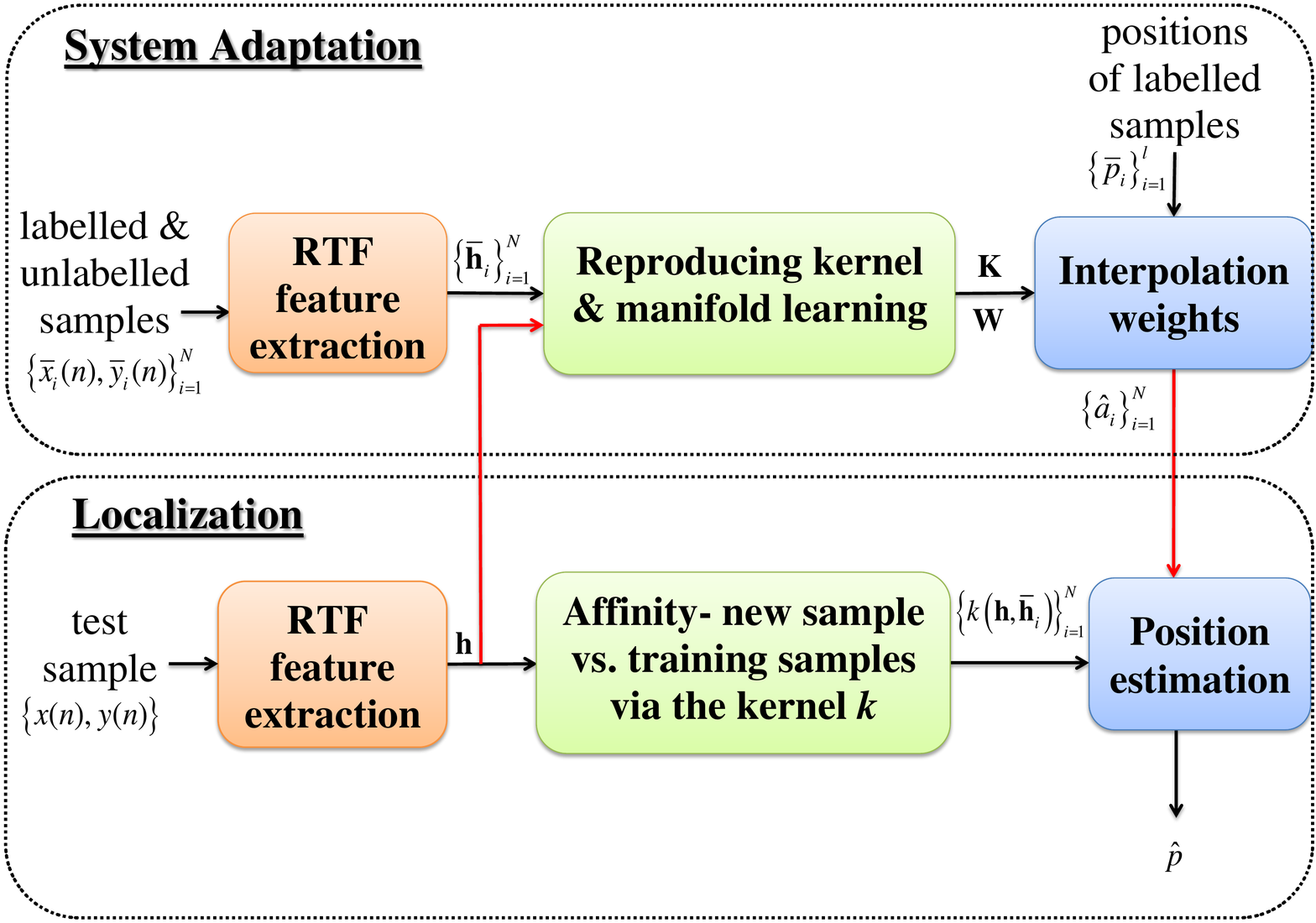}
\captionsetup{justification=centering}
\caption{Flow diagram of the proposed \ac{MRL} algorithm. The algorithm consists of two parts: system adaptation and localization. In the adaptation part, both labelled and unlabelled samples are utilized to build a data-driven model for the \acp{RTF} and relate it to the position of the source. In the localization part, the position of a new pair of measurements is estimated based on the model learnt in the adaptation stage. The newly acquired unlabelled samples in the localization stage, are propagated and utilized for system adaptation.}
\label{fig:flow}
\end{figure*}

\section{Review of Localization Based on Diffusion Mapping}
\label{sec:diffusion}
In this section we briefly review a method for semi-supervised localization that was presented in~\cite{laufer2013relative}. This method, that will be termed \ac{DDS}, is a dual stage approach based on the concepts of diffusion maps~\cite{Coifman2006,talmon2013diffusion}. In the first stage we recover the mapping between the original space $\mathbb{C}^D$ and the embedded space $\mathbb{R}^d$ which is governed by the controlling parameter, i.e. the position of the source. The second step is performing the localization by searching the neighbours of the new point among the training set in the new recovered space. Note that both the \ac{MRL} and \ac{DDS} algorithms rely on the information implied by the manifold $\mathcal{M}$. Nevertheless, there are several fundamental aspects that distinguish between the two, as will be elaborated in Section~\ref{sec:discuss}.

\subsection{Parametrization of the Manifold}
In the previous section we introduced a discrete representation of the manifold by a graph in which the training samples are the graph nodes, and the weights of the edges are defined according to the adjacency matrix $\mathbf{W}$ of \eqref{eq:W}. The adjacency graph is normalized to obtain the transition matrix $\mathbf{P}=\mathbf{D}^{-1}\mathbf{W}$, which defines a Markov process on the graph. Accordingly, $p(\mathbf{h}_i,\mathbf{h}_j)\equiv P_{ij}$ represents the probability of transition in a single Markov step from node $\mathbf{h}_i$ to node $\mathbf{h}_j$.

A nonlinear mapping of the samples into a new embedded space is obtained by spectral decomposition of the transition matrix $\mathbf{P}$. The embedding is based on a parametrization of the manifold $\mathcal{M}$, which forms an intrinsic representation of the data. We apply singular value decomposition to the transition matrix $\mathbf{P}$, and pick the $d$ principal right-singular vectors $\{\boldsymbol{\varphi}_j\}_{j=1}^{d}$ that corresponds to the $d$ largest singular values $\{\lambda_j\}_{j=1}^{d}$. The $d$ principal right-singular vectors forms the diffusion mapping of the samples into an Euclidean space $\mathbb{R}^d$, defined by:
\begin{equation}
\mathbf{\Phi}_d: \mathbf{h}_i \mapsto \left[ \lambda_1\varphi_1^{(i)}, \ldots, \lambda_d\varphi_d^{(i)}\right]^T.
\label{eq:map}
\end{equation}
where $\varphi_k^{(i)}$ denotes the $i$th entry of the vector $\varphi_k$.
Usually, $\boldsymbol{\varphi}_0$ is ignored since it is equal to a column vector of ones.

In the localization stage, the embedding should be extended, given a new \ac{RTF} sample $\mathbf{h}$, corresponding to a new pair of measurements $\{x(n), y(n)\}$ produced by unknown source from unknown location. Further spectral decomposition is unnecessary according to Nystr\"{o}m extension.  The new spectral coordinates are obtained by:
\begin{equation}
    \varphi^*_j=\frac{1}{\lambda_j}\mathbf{b}^T\boldsymbol{\varphi}_j \quad j\in \{1,\ldots,d\}
    \label{eq:Nystrom}
\end{equation}
where $\mathbf{b}$ is an affinity vector between the training set and the new test point:
\begin{equation}
    b_{i} =\exp \left\{-\frac {\|\mathbf{\bar{h}}_i-\mathbf{h}\|^2)} {\varepsilon_b}\right\}.
    \label{eq:b_kernel}
\end{equation}

\subsection{Nearest Neighbour Search on the Manifold}
\label{sec:NN}
In Section~\ref{sec:manifold} we described the structure of the acoustic manifold $\mathcal{M}$ of the \acp{RTF}. We stated that in order to properly measure affinities between \acp{RTF}, we should use the geodesic distance, which is the shortest path on the manifold. An approximation of the geodesic distance is given by diffusion distance, defined as:
\begin{align*}
D_\mathrm{Diff}^2(\mathbf{h}_i,\mathbf{h}_j)&=\|p\left(\mathbf{h}_i,\cdot\right)-p\left(\mathbf{h}_j,\cdot\right)\|_{\boldsymbol{\phi}_0}^2\\
&=\sum_{r=1}^N\left(p\left(\mathbf{h}_i,\mathbf{h}_r\right)-p\left(\mathbf{h}_j,\mathbf{h}_r\right)\right)^2/\boldsymbol{\phi}_0^{(r)}
\label{eq:diff_dist}
\end{align*}
where $\boldsymbol{\phi}_0$ is the most dominant left-singular vector of $\mathbf{P}$.

The diffusion distance incorporates information of the entire set to determine the connectivity between pairs of samples on the graph. Pairs of points who are closely related to the same subset of points in the graph, are considered close to each other and visa versa. It can be shown that the diffusion distance is equal to the Euclidean distance in the diffusion maps space when using all $N$ eigenvectors. This equivalence emphasizes the virtue of the diffusion mapping as it indicates that the mapping preserves the affinity between points with respect to the manifold. The diffusion distance can be well approximated by only the first $d$ principal eigenvectors~\cite{Coifman2006}, i.e.,
\begin{equation}
D_\mathrm{Diff}(\mathbf{h}_i,\mathbf{h}_j)\cong\|\mathbf{\Phi}_d(\mathbf{h}_i)-\mathbf{\Phi}_d(\mathbf{h}_j)\|.
\label{eq:diff_dist2}
\end{equation}

Equipped with the ability to measure distances along the manifold using the diffusion distance, we are able to properly quantify the affinities between \acp{RTF} samples. Samples which resides next to each other on the manifold, are assumed to be physically adjacent, i.e., they are likely to represent sources from close positions.
Thus, the position of a new sample can be estimated by searching for its neighbours on the manifold. Accordingly, the estimate will be formulated as a weighted sum of the positions of the labelled samples, where the weights are proportional to the corresponding diffusion distance between the new sample and each of the labelled samples:
\begin{equation}
\hat{p}=\sum_{i=1}^l\gamma \left(\mathbf{\bar{h}}_i\right)\bar{p}_i
\label{eq:diff_est}
\end{equation}
where the weights $\gamma \left(\bar{\mathbf{h}}_i\right)$ are given by:
  \begin{equation}
\gamma \left(\mathbf{h}_i\right)=\frac{\exp\left\{ -D_\mathrm{Diff}\left(\mathbf{h},\bar{\mathbf{h}}_i\right)/\varepsilon_\gamma\right\}}{\sum_{j=1}^l\exp\left\{ -D_\mathrm{Diff}\left(\mathbf{h},\bar{\mathbf{h}}_j\right)/\varepsilon_\gamma\right\} }.
\end{equation}
The \ac{DDS} procedure is summarized in Algorithm~\ref{alg:dist}.

Note that both labelled and unlabelled samples participate in the first stage, for the construction of the graph Laplacian. However, in the localization stage only the labelled samples are utilized because we rely on the labellings. Though both \ac{MRL} and \ac{DDS} algorithms have evident similarities, we show in the experimental part that the later is inferior due to its different utilization of unlabelled data.

\begin{algorithm}
\caption{\ac{DDS}}
\label{alg:dist}
\begin{algorithmic}
\STATE \underline{Diffusion Mapping}:\\
\SetKwInOut{Input}{Input}
\SetKwInOut{Output}{Output}
\Input{$N=l+u$ training points: $l$ labelled samples $\{\bar{x}_i(n),\bar{y}_i(n),\bar{p}_i\}_{i=1}^l$ and $u$ unlabelled samples $\{\bar{x}_i(n),\bar{y}_i(n)\}_{i=l+1}^N$  }
\Output{Embedding $\mathbf{\Phi}_d(\cdot)$}
\STATE \begin{enumerate}
         \item For each point estimate the corresponding \ac{RTF} $\bar{\mathbf{h}}_i$ according to \eqref{eq:rtf}.
         \item Construct the graph $\mathbf{W}$ based on $\left\{\bar{\mathbf{h}}_i\right\}_{i=1}^N$, and form the transition matrix $\mathbf{P}$.
         \item Employ singular value decomposition of $\mathbf{P}$ and obtain the singular-values $\{\mu_j\}$ and the right-singular vectors $\{\boldsymbol{\varphi}_j\}$.
         \item Construct the map $\mathbf{\Phi}_d$ according to \eqref{eq:map} to obtain an embedding that represents the intrinsic structure of manifold $\mathcal{M}$.
		\end{enumerate}
\STATE		
\STATE \underline{Localization}: \\
\Input{A new pair of measurements $\{x(n),y(n)\}$ produced by an unknown source from an unknown location}
\Output{Estimated position $\hat{p}$}
\STATE \begin{enumerate}
         \item Estimate the corresponding \ac{RTF} $\mathbf{h}$ according to \eqref{eq:rtf}.
         \item Apply Nystr\"{o}m extension according to \eqref{eq:Nystrom} to obtain the spectral coordinates of $\mathbf{h}$.
         \item Compute the approximated diffusion distance between $\mathbf{\Phi}_d(\mathbf{h})$ and each of the labelled samples $\{\mathbf{\Phi}_d(\bar{\mathbf{h}}_i)\}_{i=1}^N$, according to \eqref{eq:diff_dist2}.
         \item Estimate the new point location by \eqref{eq:diff_est} as a linear combination of the positions of the labelled samples according to distances in the diffusion mapped space.
       \end{enumerate}
\end{algorithmic}
\end{algorithm}

\section{Experimental Results}
\label{sec:results}

\subsection{Setup}
We describe the simulated setup used for conducting the experimental study. We simulated a $6\times6.2\times3$~m room, using an efficient implementation~\cite{habets2006}, of the image method~\cite{Image79}. In the room there are two microphones located at $(3,3,1)$~m and $(3.2,3,1)$~m, respectively. The source is known to be positioned at $2$~m distance with respect to the first microphone, on the same latitude. The goal is to recover the azimuth angle of the source. The initial analysis and examination of algorithms is carried out assuming that the azimuth angle of the source is ranging between $10^{\circ}\div60^{\circ}$. Then, the algorithm performance is further demonstrated on a wider range of azimuth angles between $0^{\circ}\div180^{\circ}$. Fig.~\ref{fig:Setup} illustrates the simulation setup.

\begin{figure}[ht!]
\centering
\includegraphics[width=0.5\textwidth,height=0.27\textheight]{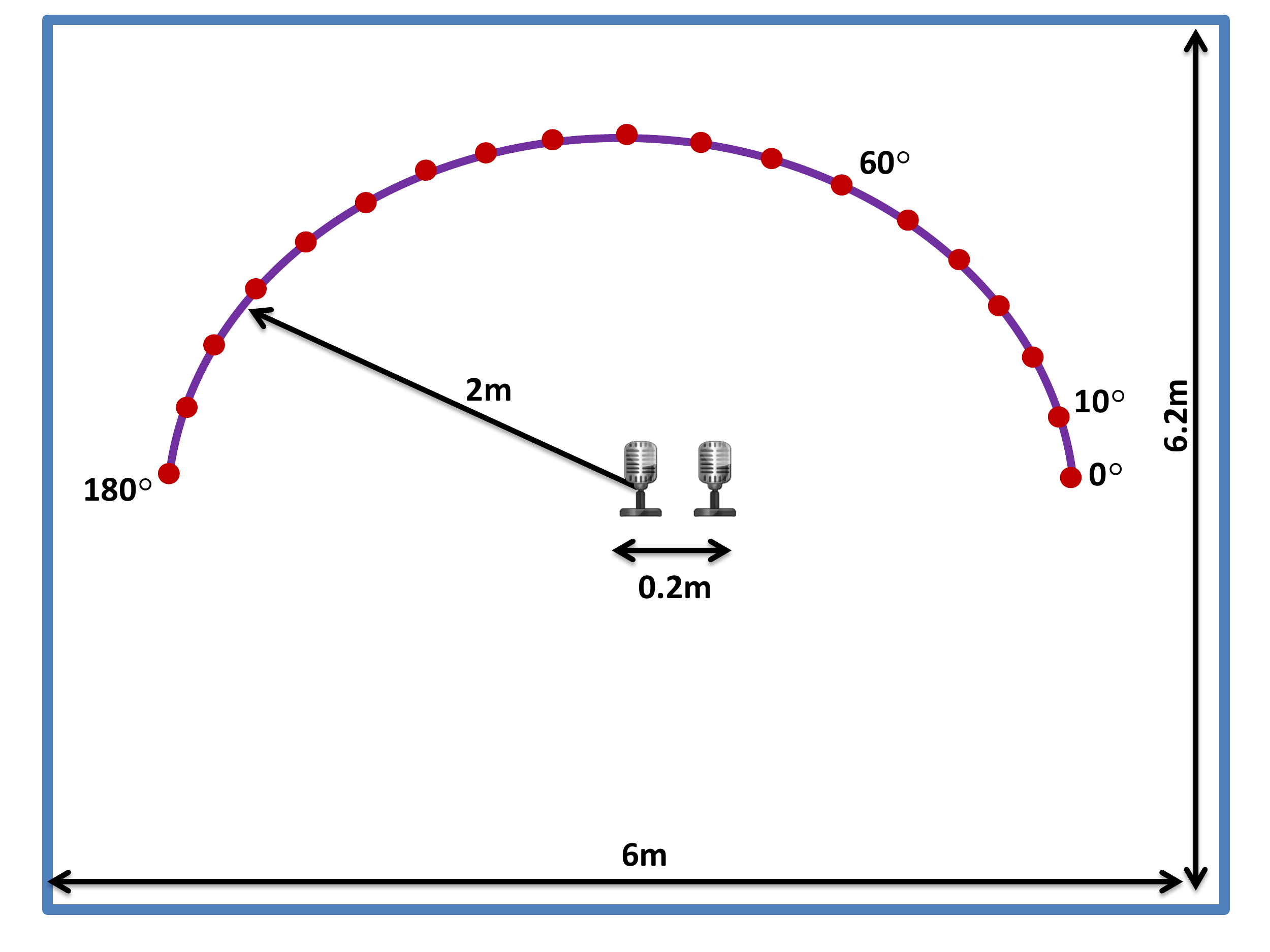}
\caption{An illustration of room setup. The purple arc marks the region where the source is assumed to be positioned. The red dots define the grid of the labelled examples.} \label{fig:Setup}
\end{figure}
For each location, we simulate a unique $3$~s speech signal, sampled at $16$~kHz. The clean speech is convolved with the corresponding \ac{AIR} and is contaminated by a WGN. This forms the measured signals in the two microphones. For each source location, the \ac{CPSD} and the \ac{PSD} are estimated with Welch's method with $0.128$~s windows and $75\%$ overlap and are utilized for estimating the \ac{RTF} in~\eqref{eq:rtf} for $D=2048$ frequency bins.

\subsection{Analysis of the Manifold}
\label{sec:analysis}
In this section we review the main results presented in~\cite{laufer2015}. We investigate the acoustic manifold of the \acp{RTF} and examine the proper distance between them that maintains physical adjacency. The analysis is carried out using a set of $N=400$ \ac{RTF} samples, corresponding to $400$ positions distributed uniformly in the specified range. Two alternative distance measures for quantifying the affinity between different \acp{RTF}, are addressed. We start with the Euclidean distance defined by:
\begin{equation}
D_\mathrm{Euc}(\mathbf{h}_i,\mathbf{h}_j)=\Vert\mathbf{h}_i-\mathbf{h}_j\Vert .
\end{equation}
The Euclidean distance is compared with the diffusion distance presented in Section~\ref{sec:NN}.

Fig.~\ref{fig:metrics&mapping}(a) depicts the Euclidean distance and the diffusion distance between each of the \acp{RTF} and a reference \ac{RTF} corresponding to $10^{\circ}$, as a function of the angle. We used moderate reverberation time of $300$~ms and $20$~dB \ac{SNR}. We observe that the monotonic behaviour of the Euclidean distance with respect to the angle is confined to approximately $3.2^{\circ}$ range. Consequently, we conclude that the Euclidean distance is meaningful only for small arcs. Thus, in general the Euclidean distance is not a good distance measure between \acp{RTF}. However it can be properly utilized when inserted into a Gaussian kernel in either the manifold regularization framework or the diffusion framework. According to its scaling parameter, the Gaussian kernel preserves small distances and suppresses large distances which are meaningless. The kernel scale should be adjusted to the distance at which monotonicity is maintained by the Euclidean distance, in order to preserve locality.

\begin{figure}[ht!]
\centering
\subfigure[]{\includegraphics[width=0.5\textwidth,height=0.3\textheight]{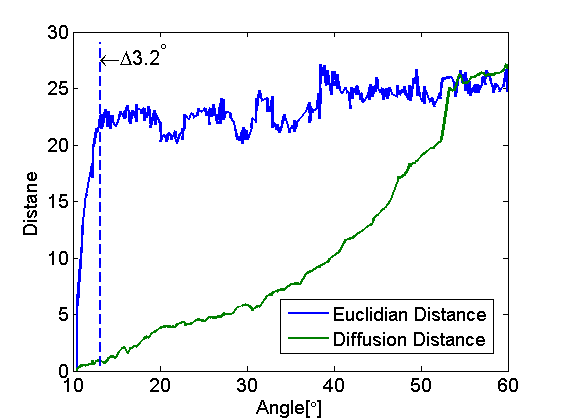}}
\subfigure[]{\includegraphics[width=0.5\textwidth,height=0.3\textheight]{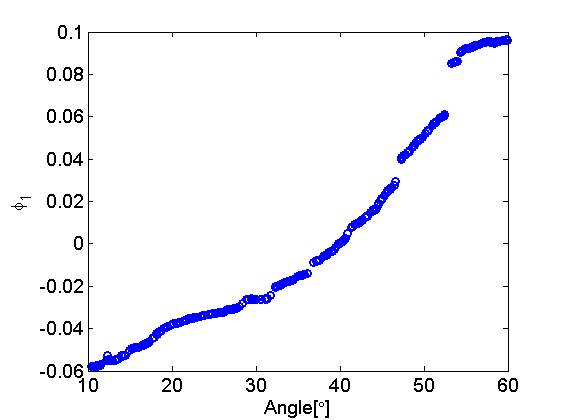}}
\caption{(a) The Euclidean distance and the diffusion distance between each of the \acp{RTF} and the \ac{RTF} corresponding to $10^{\circ}$, as a function of the angle. The dashed line shows the boundary angle until which monotonicity is preserved for the Euclidean distance. (b) Single-element diffusion mapping $\mathbf{\Phi}_1(\cdot)$.}
\label{fig:metrics&mapping}
\end{figure}

For the diffusion distance, only the first element in the mapping ($d=1$) was considered. This choice will be justified in the sequel. We can see that for almost the entire range, the diffusion distance remains monotonic with respect to the angle, indicating that it is an appropriate metric in terms of the source \ac{DOA}. Further insight into the mapping itself, is gained by plotting the single-element mapping $\mathbf{\Phi}_1(\cdot)$, as depicted in Fig.~\ref{fig:metrics&mapping}(b). We observe that the mapping corresponds well with the angle up to a monotonic distortion. Thus, the diffusion mapping successfully reveals the latent variable, namely, the position of the source. The almost perfect matching between the first element of the mapping and the corresponding angle, justifies the use of $d=1$ for estimating the diffusion distance.

To summarize, the presented results strengthen the claim on the existence of a nonlinear acoustic manifold. In small neighbourhoods around each point, the manifold is approximately flat, meaning that it resembles an Euclidean (linear) space. For larger scales the affinity between \acp{RTF} should be determined according to the geodesic distance on the manifold. The diffusion framework successfully reveals the latent variable controlling the acoustic manifold, and the diffusion distance properly reflects the distances on the manifold. These results motivate the involvement of manifold aspects in the localization process, as introduced by either the \ac{MRL} or the \ac{DDS} algorithms.

\subsection{Localization Results}
In this section we examine the ability of both \ac{DDS} and \ac{MRL} to recover the \ac{DOA} of the source. The training set consists $N=400$ representative samples distributed uniformly between $10^{\circ}\div60^{\circ}$. Among the training set, only $l=6$ samples were labelled, creating a grid with approximately $10^\circ$ distance between adjacent labelled samples, as depicted in Fig.~\ref{fig:Setup}. The performance is examined on a set of $T=120$ additional samples produced by unknown sources from unknown locations, confined to the defined range. The performance is measured according to the \ac{RMSE}, defined by:
\begin{equation}
\textrm{RMSE}= \sqrt{\frac{1}{T} \sum \limits _{i=1}^T \left\|p_i - \hat{p}_i\right\|^2}
\end{equation}
where $p$ stands for the azimuth angle of the source. To prevent the results from being dependent on a specific reflection pattern of a certain room section, we repeated the simulation with rotations of the constellation described above. The rotation angle was generated uniformly between $0^{\circ}\div360^{\circ}$. The positions of the second microphone, the training points and the test points were rotated by this angle, with respect to the first microphone. The \ac{RMSE} was averaged over $50$ rotations of the constellation.

The results of the \ac{MRL} and the \ac{DDS} algorithms are compared with that obtained by the classical \ac{GCC} algorithm~\cite{Knapp1976} for both noisy and reverberant conditions. In the first scenario we examine the algorithms' performance for different reverberation times with fixed \ac{SNR} of $20$~dB. In the second scenario the reverberation time is set to $300$~ms, and different noise levels are examined. The training set is generated with fixed \ac{SNR} level of $10$~dB. The \ac{RMSE} of the three algorithms in both scenarios, are shown in Fig.~\ref{fig:Compare_rmse}(a) and (b), respectively.

\begin{figure}[ht!]
\centering
\subfigure[]{\includegraphics[width=0.5\textwidth,height=0.3\textheight]{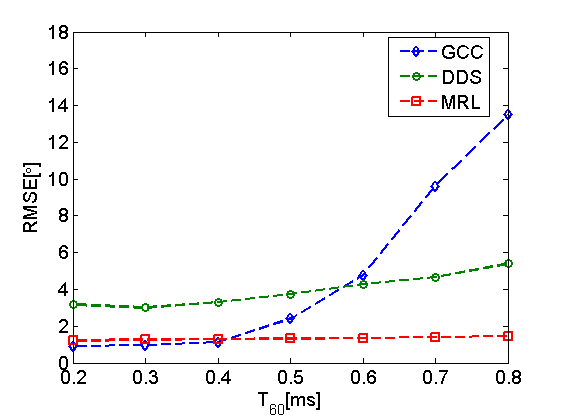}}
\subfigure[]{\includegraphics[width=0.5\textwidth,height=0.3\textheight]{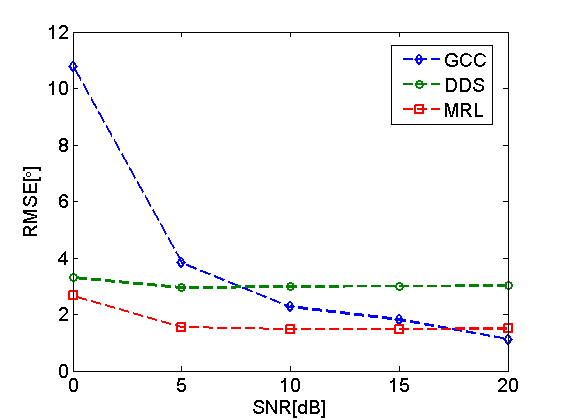}}
\caption{The RMSE of \ac{GCC}, \ac{DDS} and \ac{MRL} (a) as a function of the reverberation time (SNR=$20$~dB), and (b) as a function of \ac{SNR} ($T_{60}=300$~ms)}
 \label{fig:Compare_rmse}
\end{figure}

It can be seen in Fig.~\ref{fig:Compare_rmse}(a) that the \ac{GCC} performs well for low reverberation. However, its performance deteriorates gradually as reverberation increases, and becomes inferior compared with the performance of both the \ac{DDS} and the \ac{MRL} algorithms. In high reverberation, the \ac{GCC} is incapable of distinguishing between the direct arrival and the reflections. A misidentification of the direct path, results in a large estimation error. The proposed algorithms are more robust to reverberation, since the variations in the entire \acp{RTF} are taken in account.

Similar behaviour is observed in Fig.~\ref{fig:Compare_rmse}(b) in which different noise levels are examined. Here too, the \ac{GCC} method behaves well only in high \ac{SNR} conditions, and its performance  significantly degrades as noise level increases. When the measurements are contaminated by a significant amount of noise, the correlation between the two measurements is also very noisy, and the \ac{GCC} cannot correctly identify the peak corresponding to the direct path. On the contrary, the semi-supervised algorithms are much more robust with respect to the background noise, and most of the time obtain lower error. These type of algorithms can compensate for the information loss caused by the poor conditions, by capitalizing on the prior information inferred from the training samples.

We also observe that the \ac{MRL} approach exhibits better results compared with \ac{DDS} method. The reason for the visible gap between the \acp{RMSE} of the two algorithms is related to the different ways they utilize unlabelled data, and will be further elaborated in Section~\ref{sec:discuss}.

Finally, we examine the iterative process of the \ac{MRL} algorithm through the following sequential simulation. We used reverberation time of $500$~ms and $20$~dB \ac{SNR}. This time we examined a wider range of angles between $0^{\circ}\div180^{\circ}$. The initial adaptation was based on only $19$ labelled samples, creating a grid of $10^\circ$ distance between adjacent labelled samples, as depicted in Fig.~\ref{fig:Setup}. We conducted $9$ cycles of the sequential algorithm, each comprised of both stages of system adaptation and localization. In the localization stage, we estimated the angles of $90$ new samples from unknown locations. The total \ac{RMSE} of the all set was computed. In the following iteration, these $90$ new samples were treated as additional unlabelled data, utilized for system adaptation. The results are summarized in Fig.~\ref{fig:seq}.
\begin{figure}[ht!]
\centering
\includegraphics[width=0.5\textwidth,height=0.3\textheight]{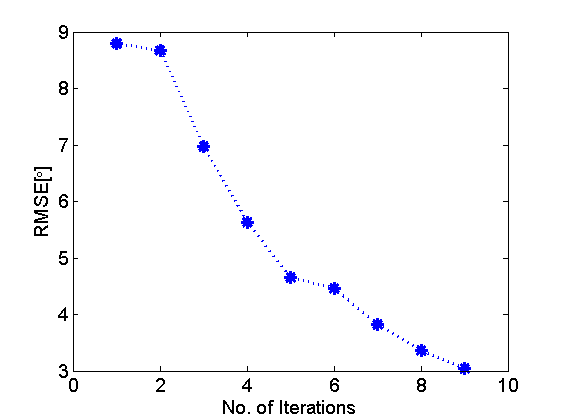}
\caption{The \ac{RMSE} of an iterative simulation of \ac{MRL} for angles in the range $0^{\circ}\div180^{\circ}$, where $90$ unlabelled points are added in each iteration. $T_{60}=500$~ms and SNR=$20$~dB}
\label{fig:seq}
\end{figure}

In this figure we observe that the \ac{RMSE} decreases as a function of the number of iterations, indicating that the unlabelled data has an important role in reducing the estimation error. However, after a considerable amount of unlabelled data is accumulated, the process stabilizes on a certain error, and additional samples are redundant.

\section{Discussion}
\label{sec:discuss}
In the previous section we demonstrated the robustness of the \ac{MRL} and the \ac{DDS} algorithms to noisy and reverberant conditions. We have also seen that the performance of the \ac{DDS} method is inferior with respect to that of the \ac{MRL} algorithm. In this section we discuss the interfacing points of both algorithms, on the one hand, and highlight the major differences between them, on the other hand.

To investigate the role of the unlabelled data in the \ac{MRL} method, we inspect the expansion weights $\mathbf{a}^*$ derived by the algorithm, as depicted in Fig.~\ref{fig:a_weights}. The blue line corresponds to the weights of $u=441$ unlabelled examples, while the red x-marks corresponds to the weights of $l=19$ labelled examples. We observe a monotonic, almost linear, behaviour of the coefficients with respect to the angle. The obtained behaviour of the \ac{MRL} coefficients, resembles the monotonic relation between the single-element diffusion mapping $\mathbf{\Phi}_1(\cdot)$ and the corresponding angle, depicted in Fig.~\ref{fig:metrics&mapping}(b). The correspondence between the two algorithms, suggests that they share similar aspects which lead to a parametrization of the manifold and recovery of the \ac{DOA} of the source.

However, we have seen that the \ac{MRL} is a better localizer compared with the \ac{DDS}. The difference between the two, is attributed to their different utilization of the unlabelled data. In the \ac{DDS} algorithm, the unlabelled data are used only in the learning phase, and the estimation merely comprises the positions of the labelled samples. In contrast, in \ac{MRL} the unlabelled data do not only take part in the recovery of the manifold, but also participate in the estimation itself, involving both labelled and unlabelled data \eqref{eq:luexpen}. Another advantage of \ac{MRL} over \ac{DDS} is that it is sequentially updated, hence, it is more suitable for on-line implementations.

\begin{figure}[ht!]
{\includegraphics[width=0.5\textwidth,height=0.3\textheight]{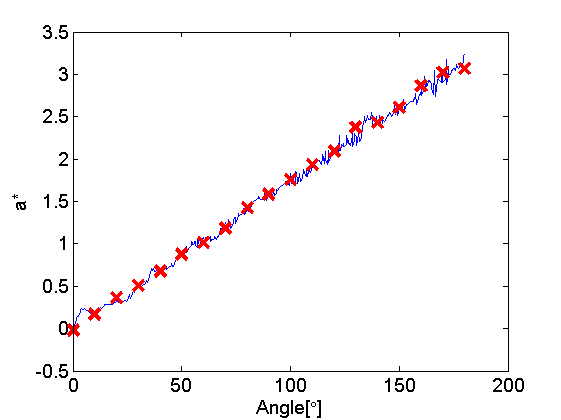}}
\caption{The estimated expansion weights $\mathbf{a}^*$ with respect to the corresponding angle.
The blue line corresponds to the weights of $u=441$ unlabelled examples, while the red x-marks corresponds to the weights of $l=19$ labelled examples.}
\label{fig:a_weights}
\end{figure}

\section{Conclusions}
\label{sec:conclusions}
A novel approach for semi-supervised localization, based on state-of-the-art manifold learning techniques, was presented. A set of representative samples in a defined room section is utilized for learning the acoustic manifold of the \acp{RTF} and building a data-driven model. Equipped with this knowledge, we find the function relating the samples and the corresponding positions by solving a regularized optimization problem in an \ac{RKHS}. Simulation results confirm the algorithm robustness in noisy and reverberant environments.

Integrating between traditional signal processing techniques and novel machine learning tools may be the key for better addressing adverse conditions, such as high noise levels and reverberations, that are the main causes for performance degradation of classical localization approaches. The current results indicate that the manifold perspective exhibits an interesting insight into the general structure of the acoustic responses and offers better solutions for common signal processing problems.

\begin{appendices}
\section{}
\label{sec:appA}
We define the integral operator on functions, associated with the kernel $k$, by the following integral transform:
\begin{equation}
[T_kf]=\int k(\mathbf{t},\mathbf{s})f(\mathbf{s})d\mathbf{s} = g(\mathbf{t}).
\end{equation}
The eigenfunctions $\{\psi_i(\cdot)\}$ and eigenvalues $\{\lambda_i\}$ of the integral operator satisfy:
\begin{equation}
[T_k\psi_i]=\int k(\mathbf{t},\mathbf{s})\psi_i(\mathbf{s})d\mathbf{s} = \lambda_i\psi_i(\mathbf{t}).
\end{equation}
According to Mercer's theorem, the kernel $k$ can be expanded by:
\begin{equation}
k(\mathbf{t}, \mathbf{s}) =\sum_i\lambda_i \psi_i(\mathbf{t})\psi_i(\mathbf{s})
\end{equation}
where the convergence is absolute and uniform.
The eigenfunctions $\{\psi_i(\cdot)\}_i$ form an orthogonal set and the \ac{RKHS} can be defined as the space of functions spanned by this set:
\begin{equation}
\mathcal{H}_k = \{ f |f(\cdot) =\sum_i a_i\psi_i(\cdot) \mbox{ and } ||f||_{\mathcal{H}_k} < \infty\}
\end{equation}
where the \ac{RKHS} norm is defined by the inner product:
\begin{equation}
\left\langle f,g \right\rangle=\left\langle\sum_i a_i \psi_i(\cdot),\sum_j b_j\psi_i(\cdot)\right\rangle
=\sum_i\frac{a_i b_i}{\lambda_j}.
\end{equation}
The reproducing property holds in this representation, since:
\begin{multline}
\left\langle f(\cdot),k_\mathbf{h}(\cdot)\right\rangle=\left\langle\sum_i a_i \psi_i(\cdot),\sum_j \lambda_j \psi_j(\mathbf{h})\psi_j(\cdot)\right\rangle\\
=\sum_i\sum_j a_i\lambda_j\psi_j(\mathbf{h})\left\langle\psi_i(\cdot),\psi_j(\cdot)\right\rangle
\stackrel{(33)}{=}\sum_i a_i\psi_j(\mathbf{h})=f(\mathbf{h})
\end{multline}

\section{}
\label{sec:appB}
\begin{theorem}
The minimizer of the optimization problem \eqref{eq:man_reg2} admits an expansion in terms of labelled and unlabelled examples:
\begin{equation}
f^*(\mathbf{h})=\sum_{i=1}^Na_i k(\bar{\mathbf{h}}_i,\mathbf{h})
\end{equation}
\end{theorem}
\begin{proof}
Any function $f \in \mathcal{H}_k$ can be uniquely decomposed into $2$ components, which one is lying in the linear subspace spanned by the kernel functions in the training examples $f_{\Vert}=\mbox{span}\left\{k(\bar{\mathbf{h}}_i,\cdot),\mbox{ } i=1,\ldots,N\right\}$ and the other is lying in the orthogonal complement $f_{\perp}$:
\begin{equation}
f=f_{\Vert}+f_{\perp}=\sum_{i=1}^Na_i k(\bar{\mathbf{h}}_i,\mathbf{h})+f_{\perp}
\end{equation}
where $\langle f_{\perp}, k(\bar{\mathbf{h}}_j,\cdot) \rangle = 0$ for all $1\leq j\leq N$.

The above orthogonal decomposition and the reproducing property together, show that the evaluation of $f$ on any training point $\bar{\mathbf{h}}_j,\mbox{ } 1\leq j\leq N$ is independent of the orthogonal component $f_{\perp}$:
\begin{align}
f(\mathbf{h}_j)&=\left\langle f(\cdot),k(\bar{\mathbf{h}}_j,(\cdot))\right\rangle\\ \nonumber
&=\left\langle \sum_{i=1}^Na_i k(\bar{\mathbf{h}}_i,\cdot)+f_{\perp},k(\bar{\mathbf{h}}_j,(\cdot)) \right\rangle\\
&=\left\langle \sum_{i=1}^Na_i k(\bar{\mathbf{h}}_i,\cdot),k(\bar{\mathbf{h}}_j,(\cdot)) \right\rangle \nonumber
=\sum_{i=1}^Na_i k(\bar{\mathbf{h}}_i,\bar{\mathbf{h}}_j)
\end{align}
Consequently, the value of the empirical terms involving the loss function and the intrinsic norm in the optimization problem (the first and the third terms, respectively), are independent of $f_{\perp}$. For the second term (the norm of $f$ in $\mathcal{H}_k$), since $f_{\perp}$ is orthogonal to $\sum_{i = 1}^N a_i k(\mathbf{h}_i,\cdot)$ and only increases the norm of $f$ in $\mathcal{H}_k$, we have
\begin{align}
\|f\|_{\mathcal{H}_k}^2&=\Big\|\sum_{i=1}^Na_i k(\bar{\mathbf{h}}_i,\mathbf{h})+f_{\perp}\Big\|_{\mathcal{H}_k}^2\\ \nonumber
&=\Big \|\sum_{i=1}^Na_i k(\bar{\mathbf{h}}_i,\mathbf{h})\Big\|_{\mathcal{H}_k}^2+\Big\|f_{\perp}\Big\|_{\mathcal{H}_k}^2\\
&\geq \Big \|\sum_{i=1}^Na_i k(\bar{\mathbf{h}}_i,\mathbf{h})\Big\|_{\mathcal{H}_k}^2 \nonumber
\end{align}
Therefore setting $f_{\perp} = 0$ does not affect the first and the third terms of \eqref{eq:man_reg2}, while it strictly decreases the second term. It follows that any minimizer $f^{*}$ of  \eqref{eq:man_reg2} must have $f_{\perp} = 0$, and therefore admits a representation: $f^*(\mathbf{h})=\sum_{i=1}^Na_i k(\bar{\mathbf{h}}_i,\mathbf{h})$.

\end{proof}
\end{appendices}


\end{document}